\documentclass[a4paper]{amsart}
\usepackage{amsmath}
\usepackage{amssymb}
\usepackage{latexsym}
\usepackage{amsfonts}
\usepackage[dvips]{graphicx}
\usepackage{color}
\usepackage{graphicx,epstopdf}

\newtheorem{thm}{Theorem}

\newtheorem{corollary}{Corollary}

\usepackage{float}
\usepackage[section]{placeins}

\renewcommand{\tilde}{\widetilde}

\usepackage{color}
\usepackage{graphicx}
\usepackage{epstopdf}

\begin{document}

\title[Burchnall-Chaundy polynomials and the Laurent phenomenon ]{Burchnall-Chaundy polynomials and the Laurent phenomenon}

\author{A.P. Veselov}
\address{Department of Mathematical Sciences,
Loughborough University, Loughborough LE11 3TU, UK  and Moscow State University, Moscow 119899, Russia}
\email{A.P.Veselov@lboro.ac.uk}

\author{R. Willox}\address{Graduate School of Mathematical Sciences, The University of
Tokyo,
3-8-1 Komaba, Meguro-ku, Tokyo 153-8914}
\email{willox@ms.u-tokyo.ac.jp}

\maketitle

\begin{abstract}
The Burchnall-Chaundy polynomials $P_n(z)$
are determined by the differential recurrence relation 
$$P_{n+1}'(z)P_{n-1}(z)-P_{n+1}(z)P_{n-1}'(z)=P_n(z)^2$$
with $P_{-1}=P_0(z)=1.$ The fact that this recurrence relation has all solutions polynomial is not obvious 
and is similar to the integrality of Somos sequences and the Laurent phenomenon.
We discuss this parallel in more detail and extend it to two difference equations
$$Q_{n+1}(z+1)Q_{n-1}(z)-Q_{n+1}(z)Q_{n-1}(z+1)=Q_n(z)Q_n(z+1)$$
and 
$$R_{n+1}(z+1)R_{n-1}(z-1)-R_{n+1}(z-1)R_{n-1}(z+1)=R^2_n(z)$$
related to two different KdV-type reductions of the Hirota-Miwa 
and Dodgson octahedral equations. As a corollary we have a new form of the Burchnall-Chaundy polynomials in terms of the initial data $P_n(0)$, which is shown to be Laurent.
\end{abstract}

\bigskip

\section{Introduction}

In the 1920s Burchnall and Chaundy \cite{BCh} discovered a remarkable sequence of polynomials satisfying the recurrence relation
\begin{equation}
\label{BCh}
P_{n+1}'(z)P_{n-1}(z)-P_{n+1}(z)P_{n-1}'(z)=P_n(z)^2
\end{equation}
with $P_{-1}=P_0(z)=1$, where $'$ means differentiation in $z$: 
$$P_1=z, \,\,P_2=\frac{1}{3} (z^3+\tau_2), \,\, P_3=\frac{1}{45}(z^6+5\tau_2 z^3+\tau_3 z-5\tau_2^2),$$
$$P_4 = \frac{1}{4725}(z^{10} + 15 \tau_2 z^7 + 7 \tau_3 z^5 -35 \tau_2\tau_3 z^2 + 175 \tau_2^3 z -\frac{7}{3} \tau_3^2 + \tau_4 z^3 + \tau_4 \tau_2 ),...$$
Note that at each step we have an additional integration constant because of the freedom in the solution of the differential equation
$$P_{n+1}(z) \to P_{n+1}(z)+c P_{n-1}(z)$$
(we are using here Adler and Moser's notation with $\tau_1=0$, see \cite{AM}).

These polynomials have been rediscovered several times, most notably by Stellmacher and Lagnese in the theory of Huygens' principle \cite{L-St}
and by Adler and Moser in the theory of rational solutions of the Korteweg-de Vries equation \cite{AM}. They appeared also in the bispectral theory of Duistermaat and Gr\"unbaum \cite{DG}, who explained their important role in the theory of monodromy-free Schr\"odinger operators and the relation to Schur polynomials with triangular Young diagrams.
The Burchnall-Chaundy polynomials $P_n(z)$ are special cases of the Schur-Weierstrass polynomials introduced by Buchstaber, Enolskii and Leykin in the theory of Klein sigma-functions  \cite{BEL} (see also Nakayashiki \cite{Nakayashiki} for the relation to KP tau functions and Sato theory).

Note that the very existence of these polynomials looks like a miracle since the relation (\ref{BCh}) can be rewritten as
$$
\frac{d}{dz}\frac{P_{n+1}}{P_{n-1}}=\frac{P_{n}^2}{P_{n-1}^2},
$$
which means that all the residues of the right-hand side must be zero (see the discussion of this in \cite{BCh} and \cite{AM}).

We would like to make the parallel with the sequence
$$
p_{n+1}p_{n-1}=p_n^2+1, \quad p_{-1}=p_0=1,
$$
which surprisingly is integer for all $n$: $1, 2, 5, 13, 34, 89,\dots$ (these in fact are nothing else but every other Fibonacci number).
This sequence is related to the cluster algebra of type $A_1^{(1)}$ and gives a nice example of the so-called Laurent phenomenon studied by Fomin and Zelevinsky:
for general initial data $p_{-1}$ and $p_0$, the solution $p_n$ is a Laurent polynomial in $p_{-1}$ and $p_0$ with integer coefficients (see \cite{CZ,FZ}).
In particular, if $p_{-1}=p_0=1$ this implies the integrality of the $p_n$ (which, of course, can be proven in an elementary way as well).
We see that the Burchnall-Chaundy sequence is a functional analogue of the same phenomenon with the role of integers $\mathbb Z$ played by the polynomial ring $\mathbb Q[z].$

This work started as an attempt to see if this conceptual parallel with the Laurent phenomenon can be made a real connection.
This led us naturally to the following {\it difference Burchnall-Chaundy equation}
\begin{equation}
\label{dBCh}
Q_{n+1}(z+1)Q_{n-1}(z)-Q_{n+1}(z)Q_{n-1}(z+1)=Q_n(z)Q_n(z+1) ,
\end{equation}
with $Q_{-1}(z)=Q_0(z)=1.$

We prove that this equation also has polynomial solutions $Q_n(z)$ with coefficients that are Laurent polynomials of the initial data 
$q_k=Q_k(0)$, such that
$$
A_n~\! Q_n(z) \in \mathbb Z [z; q_1^{\pm 1},\dots, q_{n-2}^{\pm 1}, q_{n-1}, q_{n}], \qquad A_n=\prod_{j=1}^n(2j-1)!! ,
$$
where $(2k+1)!!=1\times 3\times 5 \times\dots \times (2k+1).$
The first three polynomials have the form
$$
Q_1=z+q_1, \quad Q_2=\frac{z(z^2-1)}{3}+q_1 z^2+q_1^2z+q_2,
$$
$$
Q_3=\frac{z^2(z^2-1)(z^2-4)}{45}+\frac{2q_1z^5}{15}+\frac{q_1^2z^4}{3}+\frac{(q_1^3-q_1+q_2)z^3}{3}+\frac{(3q_1q_2-q_1^2)z^2}{3}$$
$$+(\frac{q_3}{q_1}+\frac{q_2^2}{q_1}+\frac{2q_2}{3}-\frac{q_1^3}{3}+\frac{q_1}{5})z +q_3.
$$

If we set $z=m \in \mathbb Z$, the difference Burchnall-Chaundy equation simply becomes  
a certain KdV-type reduction of the Hirota-Miwa equation, which is known to have the Laurent property {\cite{FZ}. 
This explains the Laurent property of the coefficients, 
but the proof of polynomiality of $Q_n(z)$ needs additional arguments and is related 
to specific properties of the Cauchy problem we consider.
We follow here the classical approach \cite{BCh}, \cite{AM} by using an explicit description of solutions 
in terms of Casorati determinants, which are standard in the theory of integrable systems. 
As a result we obtain an independent proof of the Laurent property for the  coefficients and a recursive way for their calculation.

We show also that these polynomials, after rescaling $R_n(z)=2^{-n(n+1)/2}Q_n(z)$,
satisfy the difference equation
\begin{equation}
\label{dBCh2}
R_{n+1}(z+1)R_{n-1}(z-1)-R_{n+1}(z-1)R_{n-1}(z+1)=R^2_n(z)
\end{equation}
with the same initial data $R_{-1}(z)=R_0(z)=1.$
This equation, which we call {\it difference Dodgson equation}, for $z=m \in \mathbb Z$ is a reduction of the 3D Dodgson octahedral equation, 
which is formally equivalent to the Hirota-Miwa equation but has different geometry.

We show that under a natural continuum limit the polynomials $Q_n(z)$ become the usual Burchnall-Chaundy polynomials $P_n(z),$ 
thus giving one more proof for their existence.

A corollary of our results is the proof that in terms of the initial data $c_n=P_n(0)$,  
$$A_n~\! P_n \in \mathbb Z [z; c_1^{\pm 1},\dots, c_{n-2}^{\pm 1}, c_{n-1}, c_{n}], \qquad  A_n=\prod_{j=1}^n(2j-1)!! ,$$
have coefficients which are Laurent polynomials in $c_k$ with integer coefficients: 
\begin{gather*}
P_1 = z + c_1 ,\,\,
P_2 = \frac{1}{3}\big(z^3 + 3 c_1 z^2 + 3 c_1^2 z + 3 c_2\big) ,\\
P_3 = \frac{1}{45}\big(z^6 + 6 c_{1} z^5 + 15 c_{1}^2 z^4 + 15 (c_{1}^3 + c_{2}) z^3 + 45 c_{1} c_{2} z^2 + 45 (\frac{c_{2}^2}{c_{1}} + \frac{c_{3}}{c_{1}}) z + 45 c_{3}\big) ,\\
\begin{split}
&P_4 = 
\frac{1}{4725}\big(z^{10} + 10 c_{1} z^9 + 45 c_{1}^2 z^8 + 15 (7 c_{1}^3 + 3 c_{2}) z^7 + 105 (c_{1}^4 + 3 c_{2} c_{1}) z^6 \\ &+ 315 (\frac{c_{2}^2}{c_{1}} + \frac{c_{3}}{c_{1}} + 2 c_{1}^2 c_{2}) z^5 + 1575 (c_{2}^2 + c_{3}) z^4 + 1575 (\frac{c_{2}^3}{c_{1}^2} + c_{1} c_{3} + \frac{c_{4}}{c_{2}} + \frac{c_{3}^2}{c_{1}^2 c_{2}} + 2 \frac{c_{3}c_{2}}{c_{1}^2} ) z^3\\ &\quad+ 4725 (\frac{c_{3}^2}{c_{1} c_{2}} + \frac{ c_{2} c_{3}}{c_{1}} + \frac{c_{1}c_{4}}{c_{2}}) z^2 + 4725 (\frac{ c_{3}^2}{c_{2}} + \frac{c_{1}^2 c_{4}}{c_{2}}) z + 4725 c_{4}\big) ,~\! \hdots
\end{split}
\end{gather*}
This form of the Burchnall-Chaundy polynomials seems to be new and demonstrates one more feature of these remarkable polynomials.

\section{Dodgson, Hirota-Miwa and Burchnall-Chaundy}

In 1866 Charles Dodgson (known to the world under the name of Lewis Carroll) proposed a very original method of computing determinants called ``condensation method" \cite{Dod}.
One can view this method as the solution of a very special Cauchy problem for the discrete {\it Dodgson octahedral equation}
\begin{equation}
\label{octa}
u_{l, m+1, n+1}u_{l,m-1,n-1}-u_{l, m+1, n-1}u_{l, m-1, n+1}=u_{l-1, m, n}u_{l+1, m, n}, 
\end{equation}
where $m,n,l \in \mathbb Z, \, m \equiv n \equiv l ({\rm mod} \, 2).$ Let $A$ be an $N\times N$ matrix for which the determinant is to be computed. Assume that $N$ is odd and consider the following Cauchy data: $u_{-1, m,n}=1,$ for all $m, n\in\mathbb{Z}$, and
$$u_{0,-N-1+2i,-N-1+2j}=A_{ij},$$ when $i,j=1,\dots N$ and $u_{0,-N-1+2i,-N-1+2j}=0$ otherwise. For even $N$ we should consider the same initial data for $l=0$ and $l=1$ respectively. The Dodgson condensation can be viewed as the evolution, according to \eqref{octa}, along the positive $l$-direction inside the characteristic cone: $\max(|m,|n|)\leq N-l-(N\!\!\mod2)$. The support of such a solution for a generic matrix will therefore be a pyramid with matrix $A$ in the base and the determinant $\det A$ at the top (see Fig.1). Outside this pyramid we can assume all $u_{l,m,n}$ for $l>0$ to be zero such that the equation is trivially satisfied.

\begin{figure}
\resizebox{\textwidth}{!}{\includegraphics[]{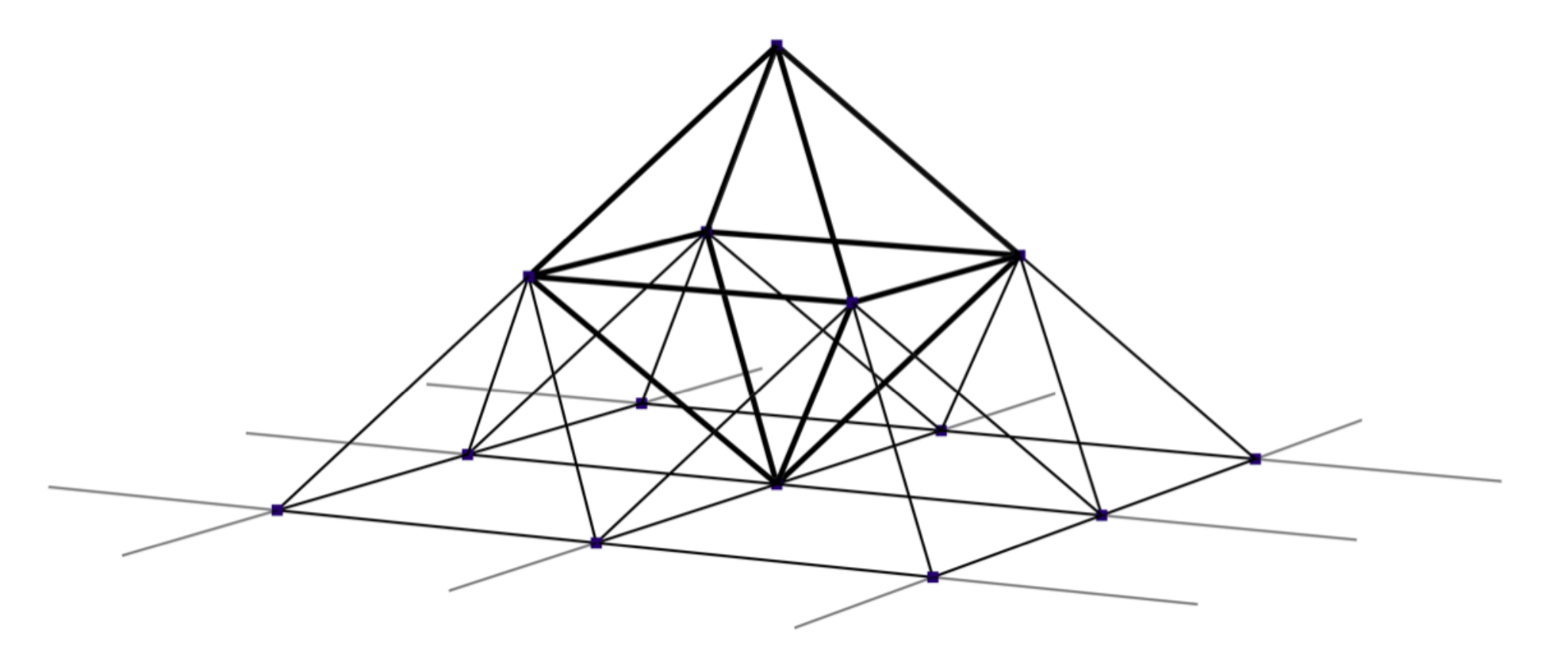}}
\caption{Dodgson's Cauchy pyramid for computing $3\times3$ determinants}
\end{figure}

In the 1980s Hirota \cite{Hirota} discovered a discrete version of the KP equation on a standard cubic lattice, which was subsequently studied by Miwa \cite{Miwa} in relation to Sato theory, and which is now known as the {\it Hirota-Miwa equation}
\begin{equation}
\label{HM}
a~\! v_{l+1, m, n}v_{l, m+1,n+1}+b~\! v_{l, m+1, n}v_{l+1, m, n+1}+c~\! v_{l, m, n+1}v_{l+1, m+1,n}=0, 
\end{equation}
where $l,m,n \in \mathbb Z$ and $a,b,c$ are arbitrary non-zero parameters. Note that the exact values of the parameters are not very important and can be changed by the gauge transformation
$$v_{l,m,n} = \Big(\frac{\tilde a}{a}\Big)^{mn}\Big(\frac{\tilde b}{b}\Big)^{ln}\Big(\frac{\tilde c}{c}\Big)^{lm}~\! {\tilde v}_{l,m,n},$$
which however will affect the Cauchy data.

Formally the  Hirota-Miwa equation may be considered as a version of the Dodgson equation if one interprets these six vertices of the cube 
as the vertices of the octahedron, see Figure \ref{cube}. 

\begin{figure}[h]
\resizebox{7cm}{!}{\includegraphics[]{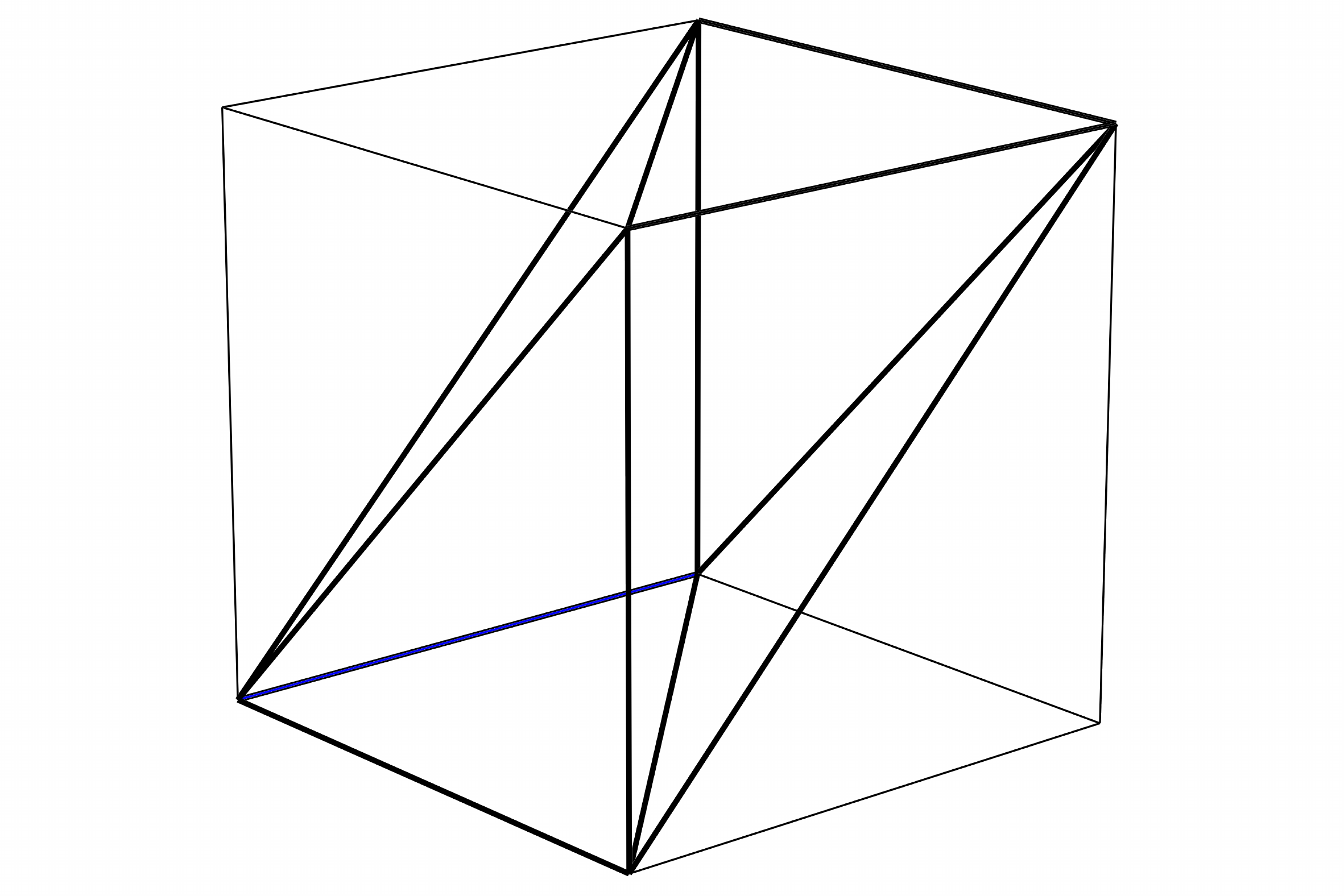}}
\caption{Support of the Hirota-Miwa equation in the cube}\label{cube}
\end{figure}

It is however clear that from a geometric point of view this is unnatural.
Indeed, the typical initial value problems for the Hirota-Miwa equation 
are very different from the Dodgson initial value problem. In particular, typical explicit solutions for the Hirota-Miwa equation are given by determinants of the same size, while for Dodgson it is crucial that the size of the matrices is shrinking when moving up in the pyramid.

This difference is clear if we compare two natural reductions of these equations. 
For the octahedral equation we consider the reduction $u_{l+1,  m,n}=u_{l-1, m, n}$ leading to the discrete Dodgson equation
\begin{equation}
\label{octa2d}
u_{m+1, n+1}u_{m-1,n-1}-u_{m+1, n-1}u_{m-1, n+1}=u^2_{m, n}, 
\end{equation}
where $m \equiv n \,({\rm mod} \, 2).$ If we extend this equation to all integers $m,n$ and replace $m$ by $z$ we come to the difference
Dodgson equation (\ref{dBCh2}).

For Hirota-Miwa a natural reduction is $v_{l+1, m, n+1}=v_{l, m,n}$ which, for $a=1, \, b=c=-1$, leads to the {\it discrete KdV equation} in Hirota bilinear form \cite{Hirota-dKdV}
\begin{equation}
\label{dKdV}
Q_{m+1, n+1}Q_{m, n-1}-Q_{m, n+1}Q_{m+1,n-1}=Q_{m,n}Q_{m+1,n},
\end{equation}
the functional version of which is the Burchnall-Chaundy equation (\ref{dBCh}). Note that the support of the equation has a domino shape, which is a cube squashed along one of the face diagonals, see Figure \ref{domino}.

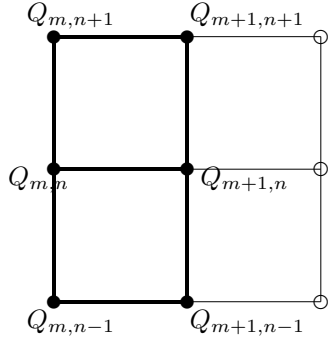
\begin{figure}[h]
\scalebox{1}{\begin{picture}(100,120)(0,0)
\linethickness{1.2pt}
\put(10,10){\line(0,1){100}}
\put(0,0){$Q_{m,n-1}$}\put(10,10){\circle*{5}}\put(10,10){\line(1,0){50}}\put(60,10){\circle*{5}}\put(61,0){$Q_{m+1,n-1}$}
\put(-7,54){$Q_{m,n}$}\put(10,60){\circle*{5}}\put(10,60){\line(1,0){50}}\put(60,60){\circle*{5}}\put(65,54){$Q_{m+1,n}$}
\put(0,116){$Q_{m,n+1}$}\put(10,110){\circle*{5}}\put(10,110){\line(1,0){50}}\put(60,110){\circle*{5}}\put(61,116){$Q_{m+1,n+1}$}
\put(60,10){\line(0,1){100}}
\linethickness{.1pt}
\put(60,10){\line(1,0){50}}
\put(60,60){\line(1,0){50}}
\put(60,110){\line(1,0){50}}
\put(110,10){\line(0,1){100}}
\put(110,10){\circle{5}}
\put(110,60){\circle{5}}
\put(110,110){\circle{5}}
\end{picture}}
\caption{Domino-type support for the discrete KdV equation.}\label{domino}
\end{figure}

Because of their common origin it is natural to expect that the equations (\ref{dBCh}) and (\ref{dBCh2}) are closely related, which is indeed the case.
\begin{thm}
The difference Burchnall-Chaundy (\ref{dBCh}) and Dodgson (\ref{dBCh2}) equations 
are equivalent on the set of initial data satisfying $\Phi_0=0,$ where
\begin{equation}
\label{phi}
\Phi_n(z):=Q_{n}(z+1)Q_{n-1}(z-1)+Q_{n-1}(z+1)Q_{n}(z-1)-2Q_{n-1}(z)Q_n(z).
\end{equation}
More precisely, if the initial data $ Q_{-1}(z),Q_0(z)$ of the Cauchy problem for the dBCh equation (\ref{dBCh}) 
satisfy the constraint
\begin{equation}
\label{constraint}
Q_{0}(z+1)Q_{-1}(z-1)+Q_{-1}(z+1)Q_{0}(z-1)-2Q_{-1}(z)Q_0(z)=0,
\end{equation}
then $R_n(z)=2^{-\frac{n(n+1}{2}}Q_n(z)$ satisfy the Dodgson relation (\ref{dBCh2}).
\end{thm}

\begin{proof}
If we modify the Dodgson equation as
\begin{equation}
\label{dBCh3}
R_{n+1}(z+1)R_{n-1}(z-1)-R_{n+1}(z-1)R_{n-1}(z+1)=2R^2_n(z),
\end{equation}
then we simply have to show that under the constraint $\Phi_0=0$ we have $Q_n(z)=R_n(z).$ 
Note that the constraint is preserved by the Burchnall-Chaundy dynamics. Indeed, assume that
\begin{equation}
\label{phi0}
\Phi_n(z)=Q_{n}(z+1)Q_{n-1}(z-1)+Q_{n-1}(z+1)Q_{n}(z-1)-2Q_{n-1}(z)Q_n(z)=0
\end{equation}
and we have two dBCh relations
\begin{equation}
\label{eq1}
Q_{n+1}(z)Q_{n-1}(z-1)-Q_{n+1}(z-1)Q_{n-1}(z)=Q_n(z-1)Q_n(z),
\end{equation}
\begin{equation}
\label{eq2}
Q_{n+1}(z+1)Q_{n-1}(z)-Q_{n+1}(z)Q_{n-1}(z+1)=Q_n(z)Q_n(z+1).
\end{equation}
Multiplying (\ref{eq1}) by $Q_n(z+1)$ and (\ref{eq2}) by $Q_n(z-1)$ we have
\begin{gather*}
\begin{split}
Q_{n+1}(z+1)Q_n(z-1)Q_{n-1}(z)-Q_{n+1}(z)Q_n(z-1)&Q_{n-1}(z+1)\\
&=Q_n(z-1)Q_n(z)Q_n(z+1),
\end{split}\\
\begin{split}
Q_{n+1}(z)Q_n(z+1)Q_{n-1}(z-1)-Q_{n+1}(z-1)&Q_n(z+1)Q_{n-1}(z)\\
&=Q_n(z-1)Q_n(z)Q_n(z+1),
\end{split}
\end{gather*}
so
\begin{gather}
Q_{n+1}(z+1)Q_n(z-1)Q_{n-1}(z)-Q_{n+1}(z)Q_n(z-1)Q_{n-1}(z+1)\nonumber
\\
=Q_{n+1}(z)Q_n(z+1)Q_{n-1}(z-1)-Q_{n+1}(z-1)Q_n(z+1)Q_{n-1}(z).
\label{con}
\end{gather}
On the other hand $Q_{n-1}(z)\Phi_{n+1}(z)-Q_{n+1}(z)\Phi_{n}(z)$ can be rewritten as
\begin{multline*}
Q_{n-1}(z)Q_{n+1}(z+1)Q_{n}(z-1)
+Q_{n-1}(z)Q_{n}(z+1)Q_{n+1}(z-1)\\
-Q_{n+1}(z)Q_{n}(z+1)Q_{n-1}(z-1)
-Q_{n+1}(z)Q_{n-1}(z+1)Q_{n}(z-1),
\end{multline*}
which is zero due to (\ref{con}).

Now in order to prove that the $R_n(z)=Q_n(z)$ satisfy (\ref{dBCh3}), we multiply relation (\ref{phi0}) by $Q_n(z)$, (\ref{eq1}) by $Q_{n-1}(z+1)$, (\ref{eq2}) by $Q_{n-1}(z-1)$ and add all of them to obtain
$$
Q_{n-1}(z)Q_{n+1}(z+1)Q_{n-1}(z-1)-Q_{n-1}(z)Q_{n+1}(z-1)Q_{n-1}(z+1)=2Q_{n-1}(z)Q^2_n(z).
$$
Dividing now by $Q_{n-1}(z)$ proves the claim.
\end{proof}
This theorem can be reformulated as follows. Consider a two by two square formed by two adjacent KdV dominos (see Figure \ref{domino}). It can be alternatively described as two horizontal dominos with relations $\Phi_n = 0$ and $\Phi_{n+1} = 0$ on them. Then the claim is that any three of these domino relations implies the fourth one and the Dodgson relation \eqref{dBCh3}.

\section{Cauchy problem for discrete equations and Laurent property}

Fomin and Zelevinsky \cite{FZ} have shown that the Hirota-Miwa and the octahedral Dodgson equations have the Laurent property 
for suitable initial data. They also considered some reductions, including the discrete KdV equation 
\begin{equation}
\alpha Q_{m+1, n+1}Q_{m, n-1} + \beta Q_{m, n+1}Q_{m+1,n-1}=Q_{m,n}Q_{m+1,n}\label{knight}\end{equation}
(which they called "the knight recurrence" and attributed to Noam Elkies).

In the case of a particular Cauchy problem for this equation with $\alpha=1, \beta=-1$ and
$$Q_{m,-1}=Q_{m,0}=1, \quad Q_{0,n}=q_n, \quad m,n \in \mathbb Z, $$
this implies that the general solution $Q_{m,n}$ is a Laurent polynomial in $q_i$ with integer coefficients.
In particular, this implies that when all $q_i=1$, all the $Q_{m,n}$ are integers as can be seen in Figure \ref{table}.
\begin{figure}[h]
\setcounter{MaxMatrixCols}{17}
\begin{equation*}
\begin{matrix}
\scriptstyle 12181&\scriptstyle  -507&\scriptstyle  -455&\scriptstyle  -91&\scriptstyle  21&\scriptstyle  5&\scriptstyle  1&\scriptstyle  21&\scriptstyle  397&\scriptstyle  6469&\scriptstyle  104145&\scriptstyle  1332565&\scriptstyle  15181325\\\scriptstyle   377&\scriptstyle  13&\scriptstyle  13&\scriptstyle  21&\scriptstyle  9&\scriptstyle  -3&\scriptstyle  1&\scriptstyle  13&\scriptstyle  149&\scriptstyle  1629&\scriptstyle  14001&\scriptstyle  115245&\scriptstyle  908245\\\scriptstyle   615&\scriptstyle  -26&\scriptstyle  -23&\scriptstyle  -4&\scriptstyle  3&\scriptstyle  2&\scriptstyle  1&\scriptstyle  8&\scriptstyle  59&\scriptstyle  350&\scriptstyle  2109&\scriptstyle  11492&\scriptstyle  52375\\\scriptstyle   249&\scriptstyle  51&\scriptstyle  5&\scriptstyle  1&\scriptstyle  1&\scriptstyle  -1&\scriptstyle  1&\scriptstyle  5&\scriptstyle  21&\scriptstyle  91&\scriptstyle  329&\scriptstyle  977&\scriptstyle  2477\\\scriptstyle   -39&\scriptstyle  -19&\scriptstyle  -7&\scriptstyle  -1&\scriptstyle  1&\scriptstyle  1&\scriptstyle  1&\scriptstyle  3&\scriptstyle  9&\scriptstyle  21&\scriptstyle  41&\scriptstyle  71&\scriptstyle  113\\\scriptstyle   -5&\scriptstyle  -4&\scriptstyle  -3&\scriptstyle  -2&\scriptstyle  -1&\scriptstyle  0&\scriptstyle  1&\scriptstyle  2&\scriptstyle  3&\scriptstyle  4&\scriptstyle  5&\scriptstyle  6&\scriptstyle  7\\\scriptstyle   1&\scriptstyle  1&\scriptstyle  1&\scriptstyle  1&\scriptstyle  1&\scriptstyle  1&\scriptstyle  1&\scriptstyle  1&\scriptstyle  1&\scriptstyle  1&\scriptstyle  1&\scriptstyle  1&\scriptstyle  1\\\scriptstyle   1&\scriptstyle  1&\scriptstyle  1&\scriptstyle  1&\scriptstyle  1&\scriptstyle  1&\scriptstyle  1&\scriptstyle  1&\scriptstyle  1&\scriptstyle  1&\scriptstyle  1&\scriptstyle  1&\scriptstyle  1\\\scriptstyle   7&\scriptstyle  6&\scriptstyle  5&\scriptstyle  4&\scriptstyle  3&\scriptstyle  2&\scriptstyle  1&\scriptstyle  0&\scriptstyle  -1&\scriptstyle  -2&\scriptstyle  -3&\scriptstyle  -4&\scriptstyle  -5\\\scriptstyle   113&\scriptstyle  71&\scriptstyle  41&\scriptstyle  21&\scriptstyle  9&\scriptstyle  3&\scriptstyle  1&\scriptstyle  1&\scriptstyle  1&\scriptstyle  -1&\scriptstyle  -7&\scriptstyle  -19&\scriptstyle  -39\\\scriptstyle   2477&\scriptstyle  977&\scriptstyle  329&\scriptstyle  91&\scriptstyle  21&\scriptstyle  5&\scriptstyle  1&\scriptstyle  -1&\scriptstyle  1&\scriptstyle  1&\scriptstyle  5&\scriptstyle  51&\scriptstyle  249\\\scriptstyle   52375&\scriptstyle  11492&\scriptstyle  2109&\scriptstyle  350&\scriptstyle  59&\scriptstyle  8&\scriptstyle  1&\scriptstyle  2&\scriptstyle  3&\scriptstyle  -4&\scriptstyle  -23&\scriptstyle  -26&\scriptstyle  615\\\scriptstyle   908245&\scriptstyle  115245&\scriptstyle  14001&\scriptstyle  1629&\scriptstyle  149&\scriptstyle  13&\scriptstyle  1&\scriptstyle  -3&\scriptstyle  9&\scriptstyle  21&\scriptstyle  13&\scriptstyle  13&\scriptstyle  377\\\scriptstyle   15181325&\scriptstyle  1332565&\scriptstyle  104145&\scriptstyle  6469&\scriptstyle  397&\scriptstyle  21&\scriptstyle  1&\scriptstyle  5&\scriptstyle  21&\scriptstyle  -91&\scriptstyle  -455&\scriptstyle  -507&\scriptstyle  12181
\end{matrix}
\end{equation*}
\caption{Solution to the discrete KdV equation with Cauchy data $Q_{0,n}=Q_{m,-1}=Q_{m,0} = 1$. The axes are given by the column of 1's ($n$ axis) and the top row of 1's ($m$ axis).}\label{table}
\end{figure}
Note that on the line $m=1$ in this figure we have the Fibonacci numbers $Q_{1,n}=F_n$, 
$$
F_{n+1}=F_{n}+F_{n-1}, \quad F_{-1}=F_{0}=1,
$$
which grow like $\varphi^n$, where $\varphi=\frac{1+\sqrt{5}}{2}$ is the 
golden ratio. On the next line we have the sequence $Q_{2,n}=G_n$ satisfying the linear recurrence,
$$
F_{n-1}G_{n+1}=F_{n+1}G_{n-1}+F_{n}G_{n}, \quad G_{-1}=G_{0}=1,
$$
with Fibonacci numbers as the coefficients, which grows like $\varphi^{2n}.$ Similarly, $|Q_{m,n}|$ for fixed $m$ grows exponentially like $\varphi^{mn}$ as $n \to \infty.$ 

In the $m$ direction the growth is much slower: for fixed $n$ the $|Q_{m,n}|$ grow like $m^{\frac{n(n+1)}{2}}.$
This indicates that the dependence on $m$ is polynomial, but to prove this for arbitrary initial data $q_n$ seems to be not easy.

Instead we prefer to use (like Burchnall and Chaundy \cite{BCh} and Adler and Moser\cite{AM}) alternative tools, standard in the theory of integrable systems, by presenting explicit determinantal expressions for the solutions. The Wronskians from \cite{AM, BCh} are naturally replaced by Casoratians in our case, which is usual in the discrete setting, see e.g. \cite{Ohta-Casorati}. Note however that in soliton KP theory, the matrices that yield determinant solutions usually have the same size, which for example in the case of the discrete KdV equation \eqref{knight} would require the sum of the coefficients $\alpha$ and $\beta$ to be equal to 1. This is not true in our case \eqref{dKdV}, for which it is known  \cite{AM,BCh} that it is essential to have matrices of growing size.

\section{Casorati determinantal formulae}

Let $Q_n(z)$ be a solution of the Cauchy problem for the difference Burchnall-Chaundy equation (\ref{dBCh})
with Cauchy data $Q_{-1}(z)=Q_0(z)=1$ and arbitrary values for $Q_n(0).$
Since the Cauchy data $Q_{-1}(z)=Q_0(z)=1$ satisfy the constraint (\ref{constraint}) 
the corresponding polynomials 
$R_n(z)$ that satisfy \eqref{dBCh2} are simply
$$R_n(z)=2^{-\frac{n(n+1}{2}}Q_n(z).$$

The following procedure is a natural difference analogue of the Adler and Moser description of Burchnall-Chaundy polynomials \cite{AM}.
Define first the sequence of polynomials $x_n(z)$  by
\begin{equation}
\label{yj}
\Delta x_n(z)=x_{n-1}(z),\quad x_0=1,
\end{equation}
where 
$$\Delta f(z):=f(z+1)-f(z).$$
They depend on some additional parameters, which are not unique. 
The generating function 
$$F(z,u):=\sum_{k=0}^{\infty}x_k(z)u^k ,$$ satisfies the 
relation
$$\Delta F=u F$$ and has the general form
\begin{align*}
F(z,u)= A(t)(1+u)^z = A(t) e^{\sum_{k=1}^{\infty}(-1)^{k+1} \frac{u^k}{k} z},
\end{align*}
where $A(t)$ is an arbitrary function of the parameters.
We fix the choice of these parameters, which will turn out to be related to the KdV flows, by choosing 
$$A(t)=e^{\sum_{k=1}^{\infty}(-1)^{k+1}t_k\frac{u^k}{k}},$$
so that 
\begin{equation}
\label{F}
F(z,t,u)=e^{\sum_{k=1}^{\infty}(-1)^{k+1}(z+t_k)\frac{u^k}{k}},
\end{equation}
leading to
$$x_0=1,\quad x_1=z+t_1,\quad x_2=\frac{1}{2}[(z+t_1)^2-(z+t_2)],\quad\dots.$$
These polynomials can be given also as the determinants (see e.g. \cite{Mac}, page 28)
\begin{equation}
\label{det1}
x_k(z)= \frac{1}{k!}
\begin{vmatrix}
z_1 & -1 & 0 & \hdots & \hdots & 0 \\
z_2 & z_1 & -2 & \hdots & \hdots & 0 \\
z_3 & z_2 & z_1 & -3 & \hdots & 0 \\
\vdots & \vdots & \vdots & \ddots & \ddots & 0 \\
z_{k-1} &  z_{k-2} & \hdots & \hdots & z_1 & 1-k \\
z_k & z_{k-1} & \hdots & \hdots & z_2 & z_1
\end{vmatrix} ,\quad \text{~} z_k = (-1)^{k+1} (z+t_k) .
\end{equation}

Let now $y_k=x_{2k-1}$ and consider the  
Casoratians $Q_n(z)=C(y_1,\dots y_n)$, where the {\it Casoratian}
$C(f_1,\dots f_n)$ of the functions $f_1(z), \dots f_n(z)$ is defined as 
\begin{equation}
\label{CW}
C(f_1,\dots f_n)=\det ||f_i(z+j-1)||, \,\, i,j=1,\dots,n.
\end{equation}
The Casoratians 
\begin{equation}
\label{Cas}
Q_k=C(y_1,\dots y_k)
\end{equation}
can be written as the determinants
\begin{equation}
\label{det3}
Q_k=
\begin{vmatrix}
x_1 & x_3 & x_5 & \hdots & \hdots & x_{2k-1} \\
1 & x_2 & x_4 & \hdots & \hdots & x_{2k-2} \\
0 & x_1 & x_3 & x_5 & \hdots & x_{2k-3} \\
0 & 1 & x_2 & x_4 & \hdots & x_{2k-4} \\
\vdots & \vdots & \vdots & \ddots & \ddots & \vdots \\
0 &  0 & \hdots & \hdots & x_{k-2} & x_k
\end{vmatrix} 
\end{equation}

\begin{thm}
The Casoratians $Q_k(z)$ given by (\ref{det3}) with $x_j(z)$ defined by (\ref{det1}) satisfy the Burchnall-Chaundy relations (\ref{dBCh}).
\end{thm}

\begin{proof}
We essentially follow the arguments of Adler and Moser \cite{AM}.

Let $\phi_1,\dots, \phi_k$ be arbitrary functions of $z$ and consider
the difference operator 
$$C_k(\chi)=C(\phi_1, \dots, \phi_k, \chi).$$
Then we have what Adler and Moser call {\it Jacobi identity}:
$$[C_k(\chi), C_{k+1}]=-(TC_k) \,\, C_{k+1}(\chi),$$
where 
$$
C_k=C(\phi_1, \dots, \phi_k),
$$ 
$T$ is the shift operator:
$$T f(z)=f(z+1),$$
and by $[A,B]=C(A,B)$ we mean the Casoratian of $A$ and $B:$
$$[A,B]=A\,\, (TB) - B\,\, (TA).$$
The proof is simple: both $[C_k(\chi), C_{k+1}]$ and $C_{k+1}(\chi)$ are difference operators 
of order $k+1$ with the same kernel generated by $\phi_1, \dots, \phi_{k+1}$.
Thus they can only differ by a factor, which is easily seen to be $-(TC_k)$.

In our case we have $\Delta^2 \phi_k=\phi_{k-1}$ and $\phi_1=z.$
It is then easy to see that this implies that if we take $\chi=1$ then
$$C_k(1)=(-1)^kC(\Delta\phi_1, \dots, \Delta\phi_k).$$
However,
$$C(\Delta\phi_1, \dots, \Delta\phi_k)=C(\Delta^2\phi_2, \dots, \Delta^2\phi_k)=C(\phi_1, \dots, \phi_{k-1})=C_{k-1},$$
and similarly: $C_{k+1}(1)=(-1)^{k+1}C_k$. Substituting this into the Jacobi identity
we have
$$
C_{k-1}\,\, (TC_{k+1})-C_{k+1} \,\,(TC_{k-1}) = C_k \,\,(TC_k),
$$
which means that $Q_n(z)=C_n$ satisfies (\ref{dBCh}).
\end{proof}

\begin{corollary}
The coefficients of the Burchnall-Chaundy polynomials $Q_k(z)$ are polynomial in the parameters $t_j$.
\end{corollary}

Here are the first few polynomials expressed in KdV parameters:
\begin{gather*}
Q_1 = z + t_1 ,\qquad Q_2 = \frac{1}{3} \big( z(z^2-1) + 3 t_1 z^2 + 3 t_1^2 z + t_1^3-t_3 \big)\\
\begin{split}
Q_3 = \frac{1}{45} \big( z^2 (z^2-1) (z^2-4) + 6 t_1 z^5 + 15 t_1^2 z^4 + (20 t_1^3 - 5 t_3 -15 t_1) z^3 +15 t_1 (t_1^3 - t_1 - t_3) z^2
 \\ + (9 t_1 - 10 t_3 + 9 t_5 - 15 t_1^2 t_3 - 5 t_1^3 + 6 t_1^5) z + t_1^6 - 5 t_3^2 - 5 t_1^3 t_3 + 9 t_1 t_5\big).
\end{split}
\end{gather*}

Now we need to find the relation between the KdV parameters $t_j$ and 
the Cauchy data $q_k=Q_k(0).$

Substituting $z=0$ in formula (\ref{det1}) we have the homogeneous polynomials
$$x_k=\frac{(-1)^{k+1}}{k}t_k+\phi_k(t_1,\dots, t_{k-1}),$$ where the weight of $x_k$ and $t_k$ is $k$:
$$x_1=t_1, \,\,\, x_2=-\frac{1}{2}t_2 + \frac{1}{2}t_1^2, \,\,\, x_3=\frac{1}{3}t_3- \frac{1}{2}t_1t_2 + \frac{1}{6}t_1^3,$$
$$x_4=-\frac{1}{4}t_4+\frac{1}{3}t_1t_3+\frac{1}{8}t_2^2-\frac{1}{4}t_1^2t_2+\frac{1}{24}t_1^4,$$
$$x_5=\frac{1}{5}t_5-\frac{1}{4}t_1t_4-\frac{1}{6}t_2t_3+\frac{1}{6}t_1^2t_3+\frac{1}{8}t_1t_2^2-\frac{1}{12}t_1^3t_2+\frac{1}{120}t_1^5,...$$

These relations can be inverted as the determinants (see \cite{Mac}, page 28)
\begin{equation}
\label{det2}
t_k=
\begin{vmatrix}
x_1 & 1 & 0 & \hdots & \hdots & 0 \\
2x_2 & x_1 & 1 & \hdots & \hdots & 0 \\
3x_3 & x_2 & x_1 & 1 & \hdots & 0 \\
\vdots & \vdots & \vdots & \ddots & \ddots & 0 \\
. &  . & \hdots & \hdots & x_1 & 1\\ 
kx_k & x_{k-1} & \hdots & \hdots & x_2 & x_1
\end{vmatrix} .
\end{equation}
Note that $t_k \in \mathbb Z[x_1, \dots, x_k]$ are polynomials with integer coefficients.

Substituting $z=0$ to (\ref{Cas}) we have the relations
$$
q_1=t_1, \,\,\,  q_2=-\frac{1}{3}t_3+\frac{1}{3}t_1^3,\,\,\, q_3=\frac{1}{5}t_1t_5-\frac{1}{9}t_3^2-\frac{1}{9}t_1^3t_3+\frac{1}{45}t_1^6,$$
$$
q_4=\frac{1}{21}(t_3-t_1^3)t_7-\frac{1}{25}t_5^2+\frac{1}{15}t_1^2t_3t_5+\frac{1}{75}t_1^5t_5-\frac{1}{27}t_1t_3^3-\frac{1}{315}t_1^7t_3+\frac{1}{4725}t_1^{10},...
$$
We can also define $q_{-1}=q_0=1.$ 

Note that the initial conditions $q_k$ are homogeneous in $t_j$ of weight $k (k+1)/2$.

\begin{thm}
The polynomial $q_k$ depends only on odd parameters $t_{2i-1},~\! (i=1,\dots,k)$ 
and has the form
$$
q_k=\frac{(-1)^{k+1}}{2k-1}q_{k-2}t_{2k-1}+\psi_k(t_1,t_3, \dots,t_{2k-3}) ,
$$
for some polynomials $\psi_k$ with rational coefficients.

The parameter $t_{2k-1}$ can be expressed in terms of $q_j$ as a Laurent polynomial with integer coefficients
$$
t_{2k-1}=(2k-1)\frac{(-1)^{k+1}q_k}{q_{k-2}}+\varphi_k(q_1,\dots,q_{k-1}) \in \mathbb Z[q_1^{\pm}, q_2^{\pm},\dots, q_{k-2}^{\pm}, q_{k-1}, q_k].
$$
\end{thm}

\begin{proof}
The generating function (\ref{F}) satisfies the differential relations
$$
\frac{\partial}{\partial t_j}F=(-1)^{j+1}\frac{u^j}{j} F,
$$
which imply that
$$
\frac{\partial x_k}{\partial t_j}=\frac{(-1)^{j+1}}{j}x_{k-j}.
$$
Differentiating the determinantal expression (\ref{det3}) with respect to $t_2$ we see that the derivative of the $j$-th column is precisely the preceding column and consequently, as the derivative of the first column is zero, that $Q_k$ does not depend on $t_2$. Similarly we see that the same is true for all even times (cf. \cite{BEL}).

This means that we can choose the even times as we want. Let us choose them in such a way that the corresponding $x_{2p}(0)=0$ for all $p\in \mathbb Z_+.$
This then defines the $t_{2p}$ as certain polynomials of odd times with rational coefficients:
$$t_2=t_1^2,\quad t_4=\frac{4}{3}t_1t_3+\frac{1}{2}t_2^2-t_1^2t_2+\frac{1}{6}t_1^4=\frac{4}{3}t_1t_3-\frac{1}{3}t_1^4, \quad \dots.$$
From the determinant (\ref{det3}) we see that in this case the $q_k=Q_k(0)$ satisfy the recurrence relation
$$q_k=(-1)^{k+1}x_{2k-1}q_{k-2}+g_k(x_1,x_3,\dots, x_{2k-3}),$$
for some polynomial $g_k$ with integer coefficients, which implies the first claim. This relation also allows us to express, 
recursively, $x_{2k-1}$ (and hence $t_{2k-1}$ using (\ref{det2})) as a Laurent polynomial of $q_1,\dots, q_{k}$ with integer coefficients.
\end{proof}

The first four expressions for the $t_{2k-1}$ have the form:
$$
t_1=q_1, \quad t_3=-3q_2+q_1^3,\quad t_5=\frac{5q_3-5q_1^3q_2+5q_2^2+q_1^6}{q_1},
$$
$$
t_7=-\frac{7q_1^2q_4+14q_2^2q_3-q_1^9q_2+7q_3^2+7q_2^4-14q_1^3q_2^3+7q_1^6q_2^2-7q_1^3q_2q_3}{q_1^2q_2}.
$$

Combining the expression (\ref{det3}) for $Q_n$ with the previous theorem we have

\begin{corollary}
Laurent phenomenon for the Burchnall-Chaundy sequence (\ref{dBCh}):
$$
A_n~\! Q_n(z) \in \mathbb Z [z; q_1^{\pm},\dots, q_{n-2}^{\pm}, q_{n-1}, q_{n}], \quad q_k=Q_k(0).
$$
\end{corollary}

The first three polynomials are given explicitly in the Introduction.

The polynomials $Q^0_n(z)$ corresponding to zero Cauchy data $q_n=Q^0_n(0)=0$ 
have to be dealt with as a special case because of the Laurent nature of the solution. Instead we can consider their expressions in terms of the $t_j$, which are polynomial, and therefore allow us to set $t_j=0$.
The corresponding polynomials admit the following explicit description.
  
Define the polynomials $Q^0_n(z)$ 
by the following recurrence relation
\begin{equation}
\label{recur}
Q^0_n(z)=\frac{1}{(2n-1)!!} \prod_{j=1}^n(z+n+1-2j)Q^0_{n-1}(z),
\end{equation}
or explicitly as
\begin{equation}
\label{expli}
Q^0_n(z)=A_n^{-1}~ z^{[\frac{n+1}{2}]}\prod_{j=1}^{n-1} (z^2-j^2)^{[\frac{n+1-j}{2}]},\quad A_n=\prod_{j=1}^n(2j-1)!!,
\end{equation}
where $(2k+1)!!=1\times 3\times 5 \times\dots \times (2k+1)$ and $[x]$ denotes the integer part of $x.$

\begin{thm}
The polynomials $Q^0_n(z)$ satisfy the difference Burchnall-Chaundy relations with zero initial data $Q^0_n(0)=0.$
\end{thm}

Proof is by a direct check.

The polynomials $Q^0_n(z)$ correspond to all $t_k=0$ and thus have Casoratian form (\ref{Cas}) with 
binomial $$x_k=\frac{z(z-1)\dots(z-k+1)}{k!}.$$
It is interesting that they also can be given as symmetric Casoratians of simple monomials. 

The {\it symmetric Casoratian} $C^*(f_1,\dots f_n)$ of the functions $f_1(z), \dots f_n(z)$ is defined as the determinant 
\begin{equation}
\label{CW*}
C^*(f_1,\dots f_n)=\det ||f_i(z+n+1-2j)||, \,\, i,j=1,\dots,n.
\end{equation}

Define the functions 
\begin{equation}
\label{fj}
f_j(x)=\frac{1}{(2j-1)!} z^{2j-1}, \,\, j=1,\dots, n.
\end{equation}

\begin{thm}
The polynomials $Q^0_n(z)$ can be given as the symmetric Casoratians
\begin{equation}
\label{CW1}
Q^0_n(z)=2^{-n(n+1)/2} ~C^*(f_1,\dots f_n)
\end{equation}
of the functions (\ref{fj}).
\end{thm}

The proof easily follows from the Vandermonde formula.

\section{Continuum limit: usual Burchnall-Chaundy polynomials}

It is well-known \cite{AMM,AM} that the Burchnall-Chaundy polynomials $P_n$ are the $\tau$-functions of the rational solutions of the 
Korteweg-de Vries equation
$$
u_{T_1}=D^3u - 6u Du, \quad D=\frac{d}{dx}
$$
and its higher analogues 
$$
u_{T_k}=D^{2k+1}u+\dots
$$
(see \cite{Miura V}, where these equations are precisely defined for scaled $u$). Namely, in a proper parametrisation
$$u(x,T_1, \dots, T_n)= -2D^2 \log P_n(x,T_1,\dots, T_n)$$
are the general rational solutions of the KdV hierarchy \cite{AMM, AM}.

Adler and Moser realised that the parameters $\tau_k$ they used (see the formulae in the Introduction) are different from the KdV times $T_k$ 
but are related to them by a non-trivial invertible polynomial transformation.

We claim that our parameters $t_{2k+1}$ are simply related to the KdV times by
the scaling $$t_{2k+1}=4^k(2k+1)T_k.$$

\begin{thm}
The continuum limit 
$$
P_n(x, t_1,t_3,\dots, t_{2n-1})= \lim_{\varepsilon \to 0}\varepsilon^{\frac{n(n+1)}{2}} Q_n(\frac{x}{\varepsilon}, \frac{t_3}{\varepsilon^{3}}, \dots, \frac{t_{2n-1}}{\varepsilon^{2n-1}}),
$$
yields the usual Burchnall-Chaundy polynomials parametrized by the scaled KdV times $t_3,\dots, t_{2n-1}.$
\end{thm}

Note that $\varepsilon^{\frac{n(n+1)}{2}} Q_n(\frac{x}{\varepsilon}, \frac{t_3}{\varepsilon^{3}}, \dots, \frac{t_{2n-1}}{\varepsilon^{2n-1}})$ is polynomial in $\varepsilon$ because of the homogeneity property of the $q$'s. The proof then follows from Miwa's results on the relation between the Hirota-Miwa equation and KP/KdV hierarchy \cite{Miwa} 
and by comparison with the normalisation of the KdV flows in \cite{Miura V}.

Since the initial data $q_k$ are homogeneous in $t_k$ they do not change during the continuum limit and we have $q_k= P_k(0) = c_k.$

\begin{corollary}
Laurent phenomenon for the usual Burchnall-Chaundy relation (\ref{BCh}):
$$
A_n~\! P_n(z) \in \mathbb Z [z; c_1^{\pm},\dots, c_{n-2}^{\pm}, c_{n-1}, c_{n}], \quad c_k=P_k(0).
$$
\end{corollary}

The explicit form of the first 4 polynomials is given in the Introduction.

\section{Concluding remarks}

The main question we are dealing with in this paper is basically what happens with the Laurent property when we replace a discrete equation by its functional difference version. Our point is that the answer depends very much on the type of the corresponding Cauchy problem.

It is known to the experts in integrable systems that the equation itself does not guarantee integrability, 
which holds only for particular initial value problems and functional class of the initial data. 
For example, for the KdV equation the initial value problem is integrable for decaying or periodic initial data \cite{Nov}, 
but for general initial data not much can be said. 
At the discrete level this fact is less visible and the importance of the choice of Cauchy problem is probably not well recognised. A discussion of the role the Cauchy problem plays with respect to the Laurent phenomenon for integrable equations can be found in \cite{Mase}.

To illustrate our point let us consider the same discrete KdV dynamics (\ref{dKdV}), but in the $m$ direction, and its 
functional version  by replacing $n$ by $x$:
$$
F_{m+1}(x+1)F_m(x-1)-F_{m+1}(x)F_m(x)-F_{m+1}(x-1)F_m(x+1)=0
$$
with $F_0(x)=1.$ It has the explicit solutions $F_m(x)=C_m\varphi^{mx}, \,\, \varphi=\frac{1+\sqrt{5}}{2}$,
but these correspond to particular Cauchy data satisfying $F_m(1)=\varphi^m F_m(0).$ 
What can one say about the analytic structure of the solutions for general Cauchy data (in particular, for $F_m(1)=F_m(0)=1$) ? So, in other words, what are the analytic functions interpolating integers on the vertical lines in Figure \ref{table} ?

It is clear that the case of the Burchnall-Chaundy polynomials is special, but the question is how special it is.
It would be interesting therefore to study other reductions, in particular period 3 reductions which correspond to the Boussinesq equation and which should be related to Schur-Weierstrass polynomials for trigonal curves \cite{BEL}. It is natural to link this with the theory of periodic Darboux/dressing chains \cite{VS, RW}.

It would also be interesting to know if the polynomials $Q_n$ play any special role
in the theory of hyperelliptic sigma-functions \cite{BEL1}.

\section{Acknowledgements}

A.P.V. is grateful to the Graduate School of Mathematical Sciences of the University of Tokyo 
for the support of his visit in April-July 2014, during which this work was done.
R.W. would like to acknowledge support from the Japan Society for the Promotion of Science,
through the JSPS grant: KAKENHI 24540204.

\end{document}